\documentclass[runningheads,a4paper]{llncs}

\usepackage{color-edits}

\addauthor{mv}{blue}
\addauthor{vs}{red}

\usepackage{amsmath,amsfonts,graphpap,amscd,mathrsfs,graphicx,lscape,amssymb}
\usepackage{epsfig,amssymb,amstext,xspace,float}
\usepackage{theorem}

\usepackage{color}              

\newif\iftr 		
\trfalse       

\iftr
\usepackage{msrtr}\msrtrno{2012-115}
\fi

\newcommand{\E}{\mathbb{E}}

\newtheorem{fact}[theorem]{Fact}

\def\squareforqed{\hbox{\rule{2.5mm}{2.5mm}}}
\def\blackslug{\rule{2.5mm}{2.5mm}}
\def\qed{\hfill\blackslug}

\def\QED{\ifmmode\squareforqed 
  \else{\nobreak\hfil   
    \penalty50                 
    \hskip1em                  
    \null                      
    \nobreak                   
    \hfil                      
    \squareforqed              
    \parfillskip=0pt           
    \finalhyphendemerits=0     
    \endgraf}                  
  \fi}

\def\blksquare{\rule{2mm}{2mm}}
\def\qedsymbol{\blksquare}
\newcommand{\bg}[1]{\medskip\noindent{\bf #1}}
\newcommand{\ed}{{\hfill\qedsymbol}\medskip}

\newenvironment{proofof}[1]{\bg{Proof of #1: }}{\ed}

\newenvironment{proofsketch}{\noindent{\it Sketch of Proof.}\hspace*{1em}}{\qed\bigskip}

\newcommand{\comment}[1]{}

\newcommand{\R}{\ensuremath{\mathbb{R}}}

\newcommand{\Q}{\ensuremath{\mathcal{Q}}}

\newcommand{\opt}{\hbox{OPT}}

\begin{document}

\title{Incentives and Efficiency in Uncertain Collaborative Environments}
\author{
Yoram Bachrach\inst{1}
\and 
Vasilis Syrgkanis\inst{2}\thanks{Work performed in part while an intern with Microsoft Research. Supported in part by ONR grant N00014-98-1-0589 and NSF grants CCF-0729006 and a Simons Graduate Fellowship.}
\and
Milan Vojnovi\'c\inst{1}
}
\institute{
Microsoft Research Cambridge, UK\\
\email{yobach,milanv@microsoft.com}
\and
Cornell University\\
\email{vasilis@cs.cornell.edu}
}

\maketitle

\begin{abstract}
We consider collaborative systems where users make contributions across multiple available projects and are rewarded for their contributions in individual projects according to a local sharing of the value produced. This serves as a model of online social computing systems such as online Q\&A forums and of credit sharing in scientific co-authorship settings. We show that the maximum feasible produced value can be well approximated by simple local sharing rules where users are approximately rewarded in proportion to their marginal contributions and that this holds even under incomplete information about the player's abilities and effort constraints. For natural instances we show almost 95\% optimality at equilibrium. When players incur a cost for their effort, we identify a threshold phenomenon: the efficiency is a constant fraction of the optimal when the cost is strictly convex and decreases with the number of players if the cost is linear. 

\end{abstract}


\iftr
 \def\App{Appendix}
\else
 \def\App{the Appendix}
\fi

\section{Introduction}

Many economic domains involve self-interested agents who participate in multiple joint ventures by investing time, effort, money or other personal resources, so as to produce some value that is then shared among the participants. Examples include traditional surplus sharing games \cite{MW96,FM99}, co-authorship settings where the wealth produced is in the form of credit in scientific projects, that is implicitly split among the authors of a paper \cite{Kleinberg2011}
and online services contexts where users collaborate on various projects and are rewarded by means of public reputation, achievement awards, badges or webpage attention (e.g. Q\&A Forums such as Yahoo! Answers, Quora, and StackOverflow  \cite{Arpita2012a,Arpita2012b,GM12,DV09,AS09,JCP09}, open source projects \cite{RLTRKR06,S06,KSL03,HO01,YK03}). 

We study the global efficiency of simple and prefixed rules for sharing the value locally at each project, even in the presence of incomplete information on the player's abilities and private resource constraints and even if players employ learning strategies to decide how to play in the game. 

The design of simple, local and predetermined mechanisms is important for applications such as sharing attention in online Q\&A forums or scientific co-authorship scenarios, where cooperative game theoretic solution concepts, that require ad-hoc negotiations and global redistribution of value, 
are less appropriate.

Robustness to incomplete information is essential in online application settings where players are unlikely to have full knowledge of the
abilities of the other players. Instead, participants have only distributional knowledge about their opponents. Additionally, public signals, such as reputation ranks, achievement boards, and history of accomplishments may result in a significant asymmetry in the beliefs about a player's
abilities. Therefore, any efficiency guarantee should be robust to the distributional beliefs and should carry over, even if player abilities are arbitrary asymmetrically distributed. 


In our main result we show that if locally at each project each player is awarded at least his marginal contribution to the value, then every equilibrium is a $2$-approximation to the
optimal outcome. This holds even when player's abilities and resource constraints are private information drawn from commonly known distributions and even when players use no-regret learning strategies to play the game. We portray several simple mechanisms that satisfy this property, such as sharing proportionally 
to the quality of the submission. Additionally, we give a generalization of our theorem, when players don't have hard constraints on their
resources, but rather have soft constraints in the form of convex cost functions. Finally, we give natural classes of instances where near optimality is achieved in equilibrium.


\textbf{Our Results.} We consider a model of collaboration where the system consists of set of players and a set of projects. Each player has a budget of time which he allocates across his projects. If a player invests some effort in some project, this results in some submission of a certain quality, which is a player and project specific increasing concave function of the effort, that depends on the player's abilities. Each project produces some value which is a monotone submodular function of the qualities of the submissions of the different participants. This common value produced by each project is then shared among the participants of the project according to some pre-specified sharing rule, e.g. equal sharing, or sharing proportionally to quality. 

\textbf{1. Marginal Contribution and Simple Sharing Rules.} We show that if each player is awarded at least his marginal contribution to the value of a project, locally, then every Nash equilibrium achieves at least half of the optimal social welfare. This holds at coarse correlated equilibria of the complete information game when
player's abilities and budget are common knowledge and at Bayes-Nash equilibria when these parameters are drawn independently from commonly known arbitrary distributions. Our result is based on showing that the resulting game is universally $(1,1)$-smooth game \cite{Roughgarden2009,Roughgarden2012,Syrgkanis2012} and corresponds to a generalization of Vetta's \cite{Vetta2002} valid utility games to incomplete information settings.  We give examples of simple sharing rules that satisfy the above condition, such as proportional to the marginal contribution or based on the Shapley value or proportional to the quality.  We show that this bound is tight for very special cases of the class of games that we study and holds even for the best pure Nash equilibrium of the complete information setting and even when the equilibrium is unique. We also analyze ranking-based sharing rules and show that they can approximately satisfy the marginal contribution condition, leading only to logarithmic in the number of players loss. 

\textbf{2. Near Optimality for Constant Elasticity.} We show that for the case when the value produced at each project is of the form $v(x)=w\cdot x^{\alpha}$ for $\alpha\in (0,1)$, where $x$ is the sum of the submission qualities, then the simple proportional to quality sharing rule achieves almost 95\% of the optimal welfare at every pure Nash equilibrium of the game, which always exists.

\textbf{3. Soft Budget Constraints and a Threshold Phenomenon.} 
When the players have soft budget constraints in the form of some convex cost function of their total effort, we characterize the inefficiency as a function of the convexity of the cost functions, as captured by the standard measure of elasticity. We show that if the elasticity is strictly greater than $1$ (strictly convex), then the inefficiency both in terms of produced value and in terms of social welfare (including player costs) is a constant independent of the number of players, that converges to $2$ as the elasticity goes to infinity (hard budget constraint case). This stands in a stark contrast with the case when the cost functions are linear, where we show that the worst-case efficiency can decrease linearly with the number of players. 

\textbf{Applications.} In the context of \emph{social computing} each project represents a specific topic on a user-generated website such as Yahoo! Answers, Quora, and StackOverflow. Each web user has a budget of time that he spends on such a web service, which he chooses how to split among different topics/questions that arise. The quality of the response of a player is dependent on his effort and on his abilities which are most probably private information.
The attention produced is implicitly split among the responders of the topic in a non-uniform manner, since the higher the slot that the response is placed in the feed, the higher the attention it gets. Hence, the website designer has the power to implicitly choose the attention-sharing mechanism locally at each topic, by strategically ordering the responses according to their quality and potentially randomizing, with the goal of maximizing the global attention on his web-service. 

Another interesting application of our work is in the context of \emph{sharing scientific credit} in paper co-authorship scenarios. One could think of players as researchers splitting their time among different scientific projects. Given the efforts of the authors at each project there is some scientific credit produced. 
Local sharing rules in this scenario translate to scientific credit-sharing rules among the authors of a paper, which is implicitly accomplished through the order that authors appear in the paper. Different ordering conventions in different communities correspond to different sharing mechanisms, with the alphabetical ordering corresponding to equal sharing of the credit while the contribution ordering is an instance of a sharing mechanism where a larger credit is rewarded to those who contributed more. 

\paragraph{}\textbf{Related Work.}
\label{sec:related}
Our model has a natural application in the context of online crowdsourcing mechanisms which were recently investigated by Ghosh and Hummel~\cite{Arpita2012a,Arpita2012b}, Ghosh and McAfee~\cite{GM12}, Chawla, Hartline and Sivan~\cite{CHS12} and Jain, Chen, and Parkes~\cite{JCP09}. All this prior work focuses on a \emph{single} project. In contrast, we consider multiple projects across which a contributor can strategically invest his effort. We also allow a more general class of project value functions. Having multiple projects creates endogenous outside options that significantly affect equilibrium outcomes.  DiPalantino and Vojnovi\' c~\cite{DV09} studied a model of crowdsourcing where users can choose exactly one project out of a set of multiple projects, each offering a fixed prize and using a ``winner-take-all'' sharing rule. In contrast, we allow the value shared to be increasing in the invested efforts and allow individual contributors to invest their efforts across multiple projects.

Splitting scientific credit among collaborators was recently studied by Kleinberg and Oren~\cite{Kleinberg2011}, who again examined players choosing a single project. 
They show how to globally change the project value functions so that optimality is achieved at some equilibrium of the perturbed game. 

There have been several works on the efficiency of equilibria of utility maximization games~\cite{Vetta2002,Goemans2004,Marden2010}, also relating efficiency with the marginal contribution property. However, this body of related work focused only on the complete information setting. For general games, Roughgarden \cite{Roughgarden2009} gave a unified framework, called smoothness, for capturing most efficiency bounds in games and showed that bounds proven via the smoothness framework automatically extend to learning outcomes. Recently, Roughgarden \cite{Roughgarden2012} and Syrgkanis \cite{Syrgkanis2012} gave a variation of the smoothness framework that also extends to incomplete information settings. Additionally, Roughgarden and Schoppmann gave a version of the framework that allows for tighter bounds when the strategy space
of the players is some convex set. In this work we utilize these frameworks to prove our results.

Our collaboration model is also related to the contribution games model of \cite{Anshelevich2010}. However, in \cite{Anshelevich2010}, the authors assume that all players get the same value from a project. This corresponds to the special case of equal sharing rule in our model. Moreover, they mainly focus on network games where each project is restricted to two participants. 

Our model is also somewhat related to the bargaining literature \cite{Hatfield2011,Kleinberg,Bateni2010}. The main question in that literature is similar to what we ask here: how should a commonly produced value be split among the participants. However, our approach is very different than the bargaining literature as we focus on simple mechanisms that use only local information of a project and not global properties of the game.

\section{Collaboration Model}
\label{sec:common-sharing}

Our model of collaboration is defined with respect to a set $N$ of $n$ players and a set $M$ of $m$ available projects. Each player $i$ participates
in a set of projects $M_i$ and has a budget of effort $B_i$, that he chooses how to distribute among his projects. Thus the strategy of player $i$ is specified by the amount of effort $x_i^j\in \R_+$ that he 
invests in project $j\in M_i$.

\paragraph{Player Abilities.} Each player $i$ is characterized by his type $t_i$, which is drawn from some abstract type space $T_i$, and which determines his abilities on different projects as well as his budget. When player $i$ invests an effort of $x_i^j$ on project $j$ this results in a submission of quality $q_i^j(x_i^j;t_i)$, which depends on his type, and which we assume to be some continuously differentiable, increasing concave function of his effort that is zero at zero. 

For instance, the quality may be linear with respect to effort $q_i^j(x_i^j)=a_i^j\cdot x_i^j$, where $a_i^j$ is some project-specific ability factor for the player that is part of his type. In the context of Q\&A forums, the effort $x_i^j$ corresponds to the amount of time spent by a participant to produce some answer at question $j$, the budget corresponds to the amount of time that the user spends on the forum, the ability factor $a_i^j$ corresponds to how knowledgeable he is on topic $j$ and $q_i^j$ corresponds to the quality of his response. 

\paragraph{Project Value Functions.}  Each project $j\in M$ is associated with a value function $v_j(q^j)$, that maps the vector of submitted qualities $q^j = (q_i^j)_{i\in N_j}$ into a produced value (where $N_j$ is the set of players that participate in the project). This function, represents the profit or revenue that can be generated by utilizing the submissions. 
In the context, of Q\&A forums $v_j(q^j)$ could for instance correspond to the webpage attention produced by a set of responses to a question. 

We assume that this value is \textit{increasing in the quality} of each submission and that it satisfies the \emph{diminishing marginal returns property}, i.e. the marginal contribution of an extra quality decreases as the existing submission qualities increase. 
More formally, we assume that the value is submodular with respect to the lattice defined on $\R^{|N_j|}$: for any $z\geq y\in \R^{|N_j|}$ (coordinate-wise) and any $w\in \R^{|N_j|}$:
\begin{equation}
v_j(w\vee z)-v_j(z)\leq v_j(w\vee y)-v_j(y),
\end{equation}
where $\vee$ denotes the coordinate-wise maximum of two vectors.
For instance, the value could be any concave function of the sum of the submitted qualities or it could  be the maximum submitted quality $v_j(q^j)=\max_{i\in N_j}q_i^j$.

\paragraph{Local Value Sharing.} We assume that the produced value $v_j(q^j)$ is shared locally among all the participants of the project, based on some predefined redistribution mechanism. The mechanism observes the submitted qualities $q^j$ and decides a share $u_i^j(q^j)$ of the project value that is assigned to player $i$, such that $\sum_{i\in N_j}u_i^j(q^j)=v_j(q^j)$.  The utility of  a player $i$ is the sum of his shares across his projects: $\sum_{j\in M_i}u_i^j(q^j)$.

In the context of Q\&A forums, the latter mechanism corresponds to a local sharing rule of splitting the attention at each topic. Such a sharing rule can be achieved by ordering the submissions according to some function of their qualities and potentially randomizing to achieve the desired sharing portions.

\paragraph{}\textbf{From Effort to Quality Space.} We start our analysis by observing that the utility of a player is essentially determined only by the submitted qualities and that there is a one-to-one correspondence between submitted quality and input effort. Hence, we can think of the players as choosing target submission qualities for each project rather than efforts. For a player to submit a quality of $q_i^j$ he has to exert effort $x_i^j(q_i^j;t_i)$, which is the inverse of $q_i^j(\cdot; t_i)$ and hence is some increasing convex function, that depends on the player's type. Then the strategy space of a player is simply be the set: 
\begin{equation}\label{eqn:strat-space}
\Q_i(t_i)= \left\{q_i=(q_i^j)_{j\in M_i}: \sum_{j\in M_i}x_i^j(q_i^j;t_i)\leq B_i(t_i)\right\}
\end{equation}
From here on we work with the latter representation of the game and define everything in quality space rather than the effort space. Hence, the utility of a player under a submitted quality profile $q$, such that $q_i\in \Q_i(t_i)$ is:
\begin{equation}\label{eqn:sharing-utility}
\textstyle u_i(q;t_i) =\sum_{j\in M_i}u_i^j(q^j).
\end{equation}
and minus infinity if $q_i\notin \Q_i(t_i)$. 

\paragraph{}\textbf{Social Welfare.} We assume that the value produced is completely shared among the participants of a project, and therefore, the social welfare is equal to the total value produced, assuming players choose feasible strategies for their type:
\begin{equation}
\textstyle SW^{t}(q)=\sum_{i\in N} u_i(q;t_i) = \sum_{j\in M} v_j(q^j) = V(q).
\end{equation}
We are interested in examining the social welfare achieved at the equilibria of the resulting game when compared to the optimal social welfare.  For a  given type profile $t$ we will denote with $\opt(t)=\max_{q\in \Q(t)}SW^t(q)$ the maximum achievable welfare.

\paragraph{}\textbf{Equilibria, Existence and Efficiency.} We examine both the complete and the incomplete information setting. In the complete information setting, the type (e.g. abilities, budget) of all the players is fixed and common knowledge. We analyze the efficiency of Nash equilibria and of outcomes that arise from no-regret learning strategies of the players when the game is played repeatedly. A Nash equilibrium is a strategy profile where no player can increase his utility by unilaterally deviating. An outcome of a no-regret learning strategy in the limit corresponds to a coarse correlated equilibrium of the game, which is a correlated distribution over strategy profiles, such that no player wants to deviate to some fixed strategy. We note that such outcomes always exist, since no-regret learning algorithms for playing games exist. When the sharing rule induces a game where each players utility
is concave with respect to his submitted quality and continuous (e.g. Shapley value) then even a pure Nash equilibrium is guaranteed to exist in our class of games, by the classic result of Rosen \cite{Rosen1965}.

In the incomplete information setting the type $t_i$ of each player is private and is drawn independently from some commonly known distribution $F_i$ on $T_i$. This defines an incomplete information game where players strategies are mappings $s_i(t_i)$, from types to (possibly randomized) actions, which in our game corresponds to feasible quality vectors. Under this assumption we quantify the efficiency of Bayes-Nash equilibria of the resulting incomplete information game, i.e. strategy profiles where players are maximizing their utility in expectation over other
player's types:
\begin{equation}
\E_{t_{-i}}[u_i(s(t))]\geq \E_{t_{-i}}[u_i(s_i',s_{-i}(t_{-i})]
\end{equation}
We note that a mixed Bayes-Nash equilibrium in the class of games that we study always exists assuming that the type space is discretized and for a sufficiently small discretization of the strategy space. Even if the strategy and type space is not discretized, a pure Bayes-Nash equilibrium 
is also guaranteed to exist in the case of soft budget constraints under minimal assumptions (i.e. type space is a convex set and utility share of a player
is concave with respect to his submitted quality and is differentiable with bounded slope) as was recently shown by Meirowitz \cite{Meirowitz2003}.

We quantify the efficiency at equilibrium with respect to the ratio of the optimal social welfare over the worst equilibrium welfare, which is denoted as the \textit{Price of Anarchy}.  Equivalently, we quantify what fraction of the optimal welfare is guaranteed at equilibrium.

\section{Approximately Efficient Sharing Rules}
\label{sec:asymmetric}

In this section we analyze a generic class of sharing rules that satisfy the property that locally each player is awarded
at least his marginal contribution to the value:
\begin{equation}
u_i^j(q^j)\geq v_j(q^j)-v^{j}(q_{-i}^j)
\end{equation}
where $q_{-i}^j$ is the vector of qualities where player $i$ submits $0$ and everyone else submits $q_i^j$.

Several natural and simple sharing rules satisfy the above property, such as sharing proportional to the marginal contribution or according to the local Shapley value.
\footnote{
The Shapley value corresponds to the expected contribution of a player to the value if we imagine drawing a random permutation and adding players sequentially, attributing
to each player his contribution at the time that he was added.} 
When the value is a concave function
of the total quality submitted, then sharing proportional to the quality: $$u_i^j(q^j)=\frac{q_i^j}{\sum_{k\in N_j}q_k^j} v_j(q^j),$$ 
satisfies the marginal contribution property. When the value is the highest quality submission, then just awarding all the value
to the highest submission (e.g. only displaying the top response in a Q\&A forum) satisfies the marginal contribution property (see Appendix for latter claims).

We show that any such sharing rule 
induces a game that achieves at least a $1/2$ approximation to the optimal social welfare, at any no-regret
learning outcome and at any Bayes-Nash equilibrium of the incomplete information setting where players'
 abilities and budgets are private and drawn from commonly known distributions. Our analysis is based on the recently introduced smoothness framework for games of incomplete information by Roughgarden \cite{Roughgarden2012} and Syrgkanis \cite{Syrgkanis2012}, which we briefly survey.

\paragraph{Smoothness of Incomplete Information Games.}
Consider the following class of incomplete information games: Each player $i$ has a type $t_i$ drawn independently
from some distribution $F_i$ on some type space $T_i$, which is common knowledge.
For each type $t_i\in T_i$ each player has a set of available actions $A_i(t_i)$. A players strategy is a function $s_i:T_i\rightarrow A_i$ that 
satisfies $\forall t_i\in T_i: s_i(t_i)\in A_i(t_i)$. The utility of a player depends on his type and the actions of all the players: $u_i:T_i\times A\rightarrow \R$. 
\begin{definition}[Roughgarden \cite{Roughgarden2012}, Syrgkanis \cite{Syrgkanis2012}]
 An incomplete information game is universally $(\lambda,\mu)$-smooth if $\forall t\in \times_i T_i$ there exists
 $a^*(t)\in \times_i A_i(t)$ such that for all $w\in \times_i T_i$ and $a\in \times_i A_i(w)$:
\begin{equation}
\sum_{i\in N} u_i(a_i^*(t),a_{-i};t_i)\geq \lambda \opt(t)-\mu \sum_{i\in N} u_i(a;w_i)
\end{equation}
\end{definition}

\begin{theorem}[Roughgarden \cite{Roughgarden2009,Roughgarden2012}, Syrgkanis \cite{Syrgkanis2012}]
 If a game is universally $(\lambda,\mu)$-smooth then
every mixed Bayes-Nash equilibrium of the incomplete information setting and every coarse correlated equilibrium of the complete information setting achieves expected social welfare at least $\frac{\lambda}{1+\mu}$ of the 
expected optimal welfare.
\end{theorem}

It is easy to observe that our collaboration model, falls into
the latter class of incomplete information games, where the action of each player is his submitted quality vector $q_i$,
and the set of feasible quality vectors depend on his private type: $\Q_i(t_i)$ as defined in Equation \eqref{eqn:strat-space}. Last the utility of a player is only a function
of the actions of other players and not directly of their types, since it depends only on the qualities that they submitted.

\begin{theorem}\label{thm:main-theorem}
The game induced by any sharing rule that satisfies the marginal contribution property is universally $(1,1)$-smooth.
\end{theorem}
\begin{proof}
Let $t, w$ be two type profiles, and let $\tilde{q}(t)\in\Q(t)$ be the quality profile that maximizes the social welfare under type profile $t$, i.e. $\tilde{q}(t)=\arg\max_{q\in \Q(t)}SW(q)$. To simplify presentation we will denote $\tilde{q}=\tilde{q}(t)$, but remind that the vector depends on the whole type profile.

Consider any quality profile $q\in \Q(w)$. By the fact that $\tilde{q}_i\in \Q_i(t_i)$ is a valid strategy for player $i$ under type profile $t_i$,
we have: 
\begin{align*}
\sum_{i\in N} u_i(\tilde{q}_i,q_{-i};t_i)  = & \sum_{i\in N} \sum_{j\in M_i} u_i^j(\tilde{q}_i^j,q_{-i}^j)
\end{align*}
By the marginal contribution property of the sharing rule we have that:
\begin{align*}
\sum_{i\in N} u_i(\tilde{q}_i,q_{-i}; t_i)  \geq & \sum_{i\in N} \sum_{j\in M_i} \left(v_j(\tilde{q}_i^j,q_{-i}^j)-v_j(q_{-i}^j)\right)=\sum_{j\in M} \sum_{i\in N_j} \left(v_j(\tilde{q}_i^j,q_{-i}^j)-v_j(q_{-i}^j)\right)
\end{align*}
Following similar analysis as in Vetta \cite{Vetta2002} for the case of complete information games, we can argue that by the diminishing marginal returns property of the value functions:
\begin{align*}
v_j(\tilde{q}_i^j,q_{-i}^j)-v_j(q_{-i}^j) \geq v_j(\tilde{q}_{\leq i}^j \vee q_{\leq i}^j,q_{>i}^j)-v_j(\tilde{q}_{<i}^j \vee q_{<i}^j,q_{\geq i}^j)
\end{align*}
Where it can be seen that the right hand side is the marginal contribution of an extra quality $\tilde{q}_i^j$ added to a larger 
vector than the vector on the left hand side. Specifically, the left hand side is the marginal contribution of $\tilde{q}_i^j$ to $q_{-i}^j$,
while the right hand side is the marginal contribution of $\tilde{q}_i^j$ to the vector $(q_{<i}^j+\tilde{q}_{<i}^j, q_{\geq i}^j)$. Summing this inequality for every player in $N_j$ we get a telescoping sum:
\begin{align*}
\sum_{i\in N_j} v_j(\tilde{q}_i^j,q_{-i}^j)-v_j(q_{-i}^j) \geq v_j(\tilde{q}^j\vee q^j)-v_j(q^j) \geq v_j(\tilde{q}^j)-v_j(q^j)
\end{align*}
Combining this with the initial inequality and using the fact that $q\in \Q(w)$, we get the desired universal $(1,1)$-smoothness property:
\begin{align*}
\sum_{i\in N} u_i(\tilde{q}_i,q_{-i};t_i)  \geq & \sum_{j\in M} v_j(\tilde{q}^j) - \sum_{j\in M} v_j(q^j) = \opt(t) - \sum_{i\in N}u_i(q;w_i)
\end{align*}
\qed\end{proof}

\begin{corollary} Under a local sharing rule that satisfies the marginal contribution property, every coarse correlated
equilibrium of the complete information setting and every mixed Bayes-Nash 
equilibrium of the incomplete information game achieves at least $1/2$ of the expected optimal social welfare.
\end{corollary}

We show that this theorem is tight for the class of games that we study and more specifically, for the natural proportional
sharing rule. The tightness holds even at pure Nash equilibria of the complete
information setting, even when all players have the same ability and even when the equilibrium is unique. Intuitively what causes inefficiency is that players prefer to congest a low value project with a very high rate of success (i.e. produces almost it's maximal value for a very small qualiy), e.g. an easy topic, rather than trying their own luck alone on a hard project that would yield very high value but would require a lot of effort to produce it. 

\begin{example}
Consider the following instance: there are $n$ players and $n$ projects. Every player participates in every project. Each player has a budget of effort of $1$ and the quality of his submission
at a project is equal to his effort. Project $1$ has value function $v_1(x) = 1-e^{-\alpha x}$, where 
$x$ is the total submitted quality. The rest of the $n-1$ projects have value function $\kappa (1-e^{-\beta x})$. We assume that value is shared proportional to the quality. We will define $\alpha, \beta$ and $\kappa$ in such a way that the unique equilibrium will be for all players to invest their whole budget on project $1$, while the optimal will be for each players' efforts to be spread out among all the projects. The uniqueness will follow from Rosen \cite{Rosen1965}, since the game is a concave
continuous utility game with convex strategy spaces.

Since value functions are continuous increasing and concave, a necessary and sufficient condition for a strategy profile to be an equilibrium is that 
the partial derivative of the share of a player with respect to his quality, is equal for all the projects that he submits a positive quality and at least as much as the partial 
derivative of his share on the remaining projects. Since we want the equilibrium to be all players putting their effort on project $1$ we need that:
\begin{align*}
\left.\left( \frac{x}{x+n-1} \left(1-e^{-a (x+n-1)}\right)\right)'\right\vert_{x=1}  > \left. \left( \kappa (1-e^{-\beta x})\right)'\right\vert_{x=0} 
\end{align*}
For $\alpha>1$ the derivative on the left hand side is at least $(n-1)/n^2$. Thus a sufficient condition for the above inequality to hold is:
$
\frac{n-1}{n^2}\geq  \kappa \beta
$.
If we let $\kappa=\frac{n-1}{\beta n^2}$ and $\alpha \rightarrow \infty$ then we satisfy the required conditions. The social welfare of the equilibrium is:
$SW(q)=1-e^{-\alpha n}\rightarrow 1$.

The optimal social welfare is at least the welfare when each player picks a different project and devotes his whole effort:
\begin{align*}
\opt\geq ~&1-e^{-\alpha}+(n-1)\cdot\kappa\cdot (1-e^{-\beta})=1-e^{-\alpha}+\frac{(n-1)^2}{\beta n^2}(1-e^{-\beta})\\
\rightarrow~& 1+\left(1-\frac{1}{n}\right)^2\frac{1-e^{-\beta}}{\beta}
\end{align*}
If we also let $\beta\rightarrow 0$ we will have:
$
\frac{\opt}{SW(q)}\rightarrow 1+\left(1-\frac{1}{n}\right)^2
$.
As $n\rightarrow \infty$, the above ratio converges to $2$.
\qed\end{example}

\subsection{Ranking Rules and Approximate Marginal Contribution}\label{sec:ranking}
An interesting, from both theoretical and practical standpoint, class of sharing rules is that of ranking rules. In a ranking sharing
scheme, the mechanism announces a set of fixed portions $a_1^j\geq \ldots \geq a_n^j$, such that $\sum_t a_t^j=1$. After the players
submit their qualities, each player is ranked based on some order that depends on the profile of qualities (e.g. in decreasing quality
order or in decreasing marginal contribution order). If a player was ranked at position $t$ then he gets a share of $a_t^j\cdot v_j(q^j)$.
Fixed reward rules capture several real world scenarios where the only way of rewarding participants is ordering them in a deterministic
manner and the designer doesn't have the freedom to award to the players arbitrary fractions of the produced value. 

We show here that although such sharing rules are quite restrictive, they are expressive enough to induce games that achieve only a logarithmic in the number of players loss in efficiency. To achieve this we show that by setting the fixed portions inversely
proportional to the position, then every player is guaranteed at least an $\log(n)$-fraction of his marginal contribution. 

It is then easy to generalize our analysis in Theorem \ref{thm:main-theorem} to show that sharing rules that award
each player a $k$-fraction of his marginal contribution induce a universally $(1/k,1/k)$-smooth game and hence
achieve a $(k+1)$-approximately optimal welfare at equilibrium.

\begin{lemma}\label{LEM:RANKING} 
By setting coefficients $a_t^j$ proportional to $\frac{1}{t}$, the game resulting from the ranking sharing rule, where submissions
are ranked with respect to the marginal contribution order, achieves a $O(\log(n))$-approximation to the optimal welfare at every coarse correlated and 
at every Bayes-Nash equilibrium.
\end{lemma}
\begin{proof}
We will show that if at each project $j$ we set $a_t=\frac{1}{t\cdot H_{n_j}}$, where $n_j=|N_j|$ and $H_n$ is the $n$-th harmonic number, then:
\begin{equation}
u_i^j(q^j)\geq \frac{1}{H_{n_j}}(v^j(q^j)-v^j(q_{-i}^j))\geq \frac{1}{H_n}(v^j(q^j)-v^j(q_{-i}^j))
\end{equation}
After showing this then the theorem follows by following a similar approach as in Theorem \ref{thm:main-theorem} to show that the 
game defined is a universally $\left(\frac{1}{H_n},\frac{1}{H_n}\right)$-smooth game, implying an $(H_n+1)$-approximation.

Observe that if under quality profile $q^j$ player $i$ is placed at position $t$, then it means that at least $t$ players have
a marginal contribution that is at least player $i$'s marginal contribution. Thus we have:
\begin{equation}
\frac{v_j(q^j)-v_j(q_{-i}^j)}{\sum_{k\in N_j} (v_j(q^j)-v_j(q_{-k}^j))}\leq \frac{1}{t}
\end{equation}
Thus we conclude that the share of a player is at least $1/H_{n_j}$ of his share under the sharing rule that splits the value proportional to the 
marginal contribution. By Lemma \ref{lem:marginal-contr} we know that the latter share is at least the marginal contribution of the player to the value. 
Combining the two we get the desired property.
\qed\end{proof}

If the value is a function of the sum of the quality of submissions then a similar guarantee is achieved if submissions are ordered in decreasing 
quality.

\section{Almost Optimality for Uniformly Hard Projects}
\label{sec:restricted}

In this section we identify a natural subclass of value functions for which the social welfare at equilibrium is a much higher approximation
to the optimal welfare, achieving almost $95\%$ of the optimal. We start our quest, by observing that the crucial factor that led
to the tight lower bound presented in the previous section, is that different projects can have a very different 
rate of success: the percentage increase in the output for a percentage increase in the input was completely different for different projects and at
different qualities within a project. 
This discrepancy in the output sensitivity was the main force driving the lower bound. In this section we examine a broad class of functions
that don't allow for such discrepancies. 

The standard economic measure that captures the sensitivity of the output of a function with respect to a change in its input is that of
elasticity.
\begin{definition}\label{defn:elasticity} 
The elasticity of a function $f(x)$ is defined as:
$
\textstyle \epsilon_{f}(x) = \left|\frac{f'(x) x}{f(x)}\right|.
$
\end{definition}
One can show formally that the above parameter of a function has a one-to-one correspondence with the ratio of the percentage change in the output
for a percentage of change in the input. Intuitively, projects whose value has the same and constant elasticity have the same and uniform difficulty, though not necessarily the same importance. 
Based on this reasoning, we examine the setting where all project value functions are functions of the total quality of submissions and have constant elasticity $\alpha$. It can be easily seen that such functions will take the form $v_j(q^j) = w_j \cdot Q_j^{\alpha}$, where $Q_j=\sum_{i\in N_j}q_i^j$. The coefficient $w_j$ can be project specific, and will correspond to the importance of a project. For such value functions we prove that the proportional to the quality sharing mechanism achieves social welfare at any pure Nash equilibrium of the complete information setting that is almost optimal. 

We point that 
our class of games always possess a pure Nash equilibrium, since they are games defined on a convex strategy space, with continuous and concave utilities and hence the existence is implied by Rosen \cite{Rosen1965}.

\begin{theorem}\label{thm:fractional-exponents}\label{THM:FRACTIONAL}
Suppose that the project value functions are of the form $v_j(q^j) = w_j\cdot Q_j^\alpha$, for $0\leq \alpha \leq 1$ and $w_j > 0$. Then, the proportional to the quality sharing rule achieves social welfare at least $\frac{2^{1-\alpha}}{2-\alpha}\geq 0.94$ of the optimal social welfare at every pure Nash equilibrium of the complete information game it defines. 
\end{theorem}

Our analysis is based on the local smoothness framework of Roughgarden and Schoppmann \cite{Roughgarden2010}, which allows deriving tighter bounds for games where the strategy space is continuous and convex and where the utilities are continuous and differentiable. 
It is easy to see that the strategy spaces in our setting $\Q_i(t_i)$ as defined in equation \eqref{eqn:strat-space} are convex, by the convexity of the functions $x_i^j(\cdot;t_i)$.\footnote{
If the effort required to produce two quality vectors $q_i, \hat{q_i}$ is at most $B_i(t_i)$, then the effort 
required to produce any convex combination of them is also at most $B_i(t_i)$} Additionally, it is easy to check that the utilities of the players under the proportional sharing rule are going to be continuous and differentiable at any point, except potentially at $0$.  

However, in \App~\ref{sec:proof-of-fractional} we show that in equilibrium no project receives $0$ total quality with positive probability.
In Appendix \ref{sec:local-smoothness} we show that this relaxed condition is sufficient to apply the local smoothness framework to pure Nash equilibria. Alternatively, we can bypass this technicality by assuming there are exclusive players who participate only at a specific project and always invest an $\epsilon$ amount of effort. Making the latter assumption, we can use the local smoothness
framework in its full generality and our conclusion in Theorem \ref{thm:fractional-exponents} carries over to correlated equilibria of the game (outcomes of no-swap regret learning strategies). 

\section{Soft Budget Constraints}
\label{sec:production-costs}

So far we have analyzed the case where players have a hard constraint on their effort, e.g. hard time constraint. In this section we relax
this assumption and study the case when instead of the budget constraint, a player incurs a cost that is a convex function of
the total effort he exerts, corresponding to a soft budget constraint on his effort. We exhibit an interesting threshold phenomenon in the inefficiency of 
the setting: if cost is linear in the total effort then the inefficiency can grow linearly with the number
of participants. However, when effort cost is strictly convex, then the inefficiency can be at most a constant independent
of the number of participants. 

More formally, we will assume that 
each player has a cost function $c_i(x;t_i)$ that determines his cost when he exerts a total effort of $x$. This cost
function is also dependent on his private type $t_i$. The total exerted effort can be expressed with respect to the quality of submission as $X_i(q_i;t_i)=\sum_{j\in M_i}x_i^j(q_i^j;t_i)$. Thus a  player's  utility as a function of the profile of chosen qualities is:
\begin{equation}
\textstyle u_i(q)=\sum_{j\in M_i}u_i^j(q^j) - c_i\left(X_i(q_i;t_i); t_i\right). 
\end{equation}
%
Unlike the previous section, the social welfare is not the value produced. Instead:
$$
\textstyle SW^t(q)=\sum_{i\in N} u_i(q;t_i)= \sum_{j\in M} v_j(q^j) - \sum_{i\in N} c_i\left(X_i(q_i;t_i); t_i\right) = V(q)-C^t(q).
$$
We refer to $V(q)$ as the production of an outcome $q$ and to $C^t(q)$ as the cost.
We first show that when a sharing rule that satisfies the marginal contribution property is used, the production plus the social welfare at equilibrium is at least the value of the optimal social welfare. We then use this result
to give bounds on the equilibrium efficiency parameterized by the convexity of the cost functions.
\begin{lemma}\label{thm:general_effort_costs}
Consider the game induced by any sharing rule that satisfies the marginal contribution property and where players have soft budget constraints. Then the expected social welfare plus the expected production at any coarse correlated equilibrium of the complete information setting and at any Bayes-Nash equilibrium of the incomplete information setting, is at least
the expected optimal social welfare.
\end{lemma}
\begin{proofsketch}
Consider two type profile $t, w$ and let $\tilde{q}$ be the optimal strategy profile for type profile $t$. Let 
$q\in \Q(w)$. Similarly to Theorem \ref{thm:main-theorem} we get:
\begin{align*}
 \sum_{i\in N} u_i(\tilde{q}_i,q_{-i};t_i) =~& \sum_{i\in N} \sum_{j\in M_i} u_i^j(\tilde{q}^j,q_{-i}^j)-\sum_{i\in N}c_i(X_i(\tilde{q}_i;t_i);t_i)\\
  \geq~& V(\tilde{q})-V(q)-C^t(\tilde{q}) = SW^t(\tilde{q})-V(q)
\end{align*}
The latter doesn't imply formally that the game is smooth under existing definitions of smoothness \cite{Roughgarden2012,Syrgkanis2012}. However, using similar random sampling techniques as the ones in \cite{Roughgarden2012,Syrgkanis2012}, we can show that although $\tilde{q}_i$
is not a valid deviation for a player in the Bayesian game (since it depends on the whole type profile which he doesn't know), a player can simulate this deviation by random sampling other
players types and then performing the deviation corresponding to the random sample of types. 
\end{proofsketch}

We use the latter result to derive efficiency bounds for both production and social welfare. 
We assume that the sharing rule used induces a utility share that is a concave function
of a players submission quality and such that a player's share at $0$ quality is $0$. More formally, we assume that: $g(x)=u_i^j(x, q_{-i}^j)$
is concave, continuously differentiable and $g(0)=0$. We call such sharing rules \textit{concave sharing rules}. It is easy to see that the proportional to quality sharing rule is a concave sharing rule when the value is concave in the total quality. 
For general value functions, it is also easy to see that the Shapley sharing rule is also a concave sharing rule.

We show efficiency bounds parameterized by the convexity of the cost functions,  
%
using the elasticity of the cost function 
as the measure of convexity.
An increasing convex function that is zero at zero, has an elasticity of at least $1$. We will quantify the inefficiency in our game as a function of how far from $1$ the elasticity of the cost functions are. For instance, $c_i(x)=\kappa\cdot  x^{1+a}$ 
has elasticity $1+a$.

\begin{theorem}\label{thm:production-costs}\label{THM:PRODUCTION-COSTS}
If a concave sharing rule is used and the elasticity of the cost functions is at least $1+\mu$ then: i) the expected social welfare at any coarse correlated equilibrium of the complete information setting and at any Bayes-Nash equilibrium of the incomplete information setting is at least $\frac{\mu}{1+2\mu}$ of the optimal, ii) the total value produced in equilibrium is at least $\frac{1}{2}\frac{\mu}{1+\mu}$ of the value produced at the social welfare maximizing outcome.
\end{theorem}
\begin{proof}
We will focus on the case of pure Bayes-Nash equilibria. The proof for correlated equilibria of the complete
information setting and for mixed Bayes-Nash equilibria is similar. We will prove only the first part
of the theorem and defer to the appendix the second part which follows a similar approach.

We will prove that if $q(t)$ is a mixed Bayes-Nash Equilibrium then 
\begin{equation}
(1+\mu)\E_{t}[C(q(t))]\leq\E_{t}[V(q(t))]
\end{equation}
Then the theorem will follow directly by combining the above with Lemma \ref{thm:general_effort_costs}.

Fix a player $i$ and his type $t_i$. For ease of presentation let $q_i=q_i(t_i)$ be his equilibrium strategy, $X_i=X_i(q_i;t_i)$ his total effort and $c_i(x)=c_i(x;t_i)$ his cost function. A player's expected utility conditional on his type $t_i$ is:
\begin{equation*}
\E_{t_{-i}}\left[u_i(q_i,q_{-i}(t_{-i}))\right] = \sum_{j\in M_i} \E_{t_{-i}}\left[u_i^j(q_i^j,q_{-i}^j(t_{-i}))\right] - c_i(X_i)
\end{equation*}
By the first order conditions, for $q_i$ to be maximizing player $i$'s expected utility it must be that for any $j\in M_i$ such that $q_i^j>0$:
\begin{equation*}
\frac{\partial c_i(X_i)}{\partial q_i^j} = \frac{\partial \E_{t_{-i}}[u_i^j(q_i^j,q_{-i}^j(t_{-i})]}{\partial q_i^j} \implies
 c_i'(X_i) = \frac{1}{(x_i^j(q_i^j))'} \frac{\partial \E_{t_{-i}}[u_i^j(q_i^j,q_{-i}^j(t_{-i})]}{\partial q_i^j}
\end{equation*}
Otherwise player $i$ would want to increase or decrease his submission at project $j$. 
Multiplying the above by $X_i$ we get:
\begin{equation*}
X_i\cdot c_i'(X_i) = \sum_{j\in M_i}x_i^j(q_i^j)  \cdot c_i'(X_i) = \sum_{j\in M_i}\frac{x_i^j(q_i^j)}{(x_i^j(q_i^j))'} \frac{\partial \E_{t_{-i}}[u_i^j(q_i^j,q_{-i}^j(t_{-i})]}{\partial q_i^j} 
\end{equation*}
By convexity of $x_i^j(q_i^j)$ we  have $x_i^j(q_i^j)\leq q_i^j\cdot (x_i^j(q_i^j))'$:
\begin{equation}
X_i\cdot c_i'(X_i) \leq \sum_{j\in M_i}q_i^j\cdot \frac{\partial \E_{t_{-i}}[u_i^j(q_i^j,q_{-i}^j(t_{-i}))]}{\partial q_i^j} 
\end{equation}
Since $u_i^j(x,q_{-i}^j)$ is a concave function, the expectation of it over $q_{-i}^j$ is also 
a concave function of $x$ that is $0$ at $0$. Thus the quantity on the right hand side is of the form $x\cdot g'(x)$ 
for a concave function $g(x)$ with $g(0)=0$. Thus it is at most $g(x)$.
\begin{equation*}
X_i\cdot c_i'(X_i) \leq \sum_{j\in M_i} \E_{t_{-i}}[u_i^j(q_i^j,q_{-i}^j(t_{-i}))]
\end{equation*}
%
By the lower bound on the elasticity of the cost functions we get:
\begin{align*}
(1+\mu) c_i(X_i)\leq~& X_i\cdot c_i'(X_i) \leq \sum_{j\in M_i} \E_{t_{-i}}[u_i^j(q_i^j,q_{-i}^j(t_{-i}))]
\end{align*}
Taking expectation over player $i$'s type and summing over all players:
\begin{equation*}
(1+\mu)\E[C(q(t))]\leq \E_{t}\left[\sum_i \sum_{j\in M_i} u_i^j(q^j(t))\right]= \E_{t}\left[\sum_{j\in M} v_j(q^j(t))\right]
\end{equation*}
\qed\end{proof}


Observe that from this theorem, we obtain that as long as $\mu>0$, the efficiency of any Nash equilibrium, is a constant independent of the number of players. For instance, if the cost is a quadratic function of the total effort then the
social welfare at equilibrium is a $3$-approximation to the optimal and the produced value is a $4$-approximation 
to the value produced at the welfare-maximizing outcome. 

The budget constraint case that we studied in previous sections can be seen as a limit of a family of convex functions that converge to a limit function of the form $c_i(x_i)=0$ if $x_i<B_i$, and $\infty$ otherwise. Such a limit function can be thought of as a convex function with infinite elasticity. 
Observe that if we take the limit as $\mu\rightarrow \infty$ in the theorem above, then we get that the social welfare at equilibrium is at least half the optimal social welfare, which matches our analysis in the previous section.

A corner case is that of linear cost functions, where our Theorem gives no meaningful upper bound. In fact as the following example shows, when cost functions are linear then the inefficiency can grow linearly with the number of agents. 

\begin{example}Consider a single project with value $v(Q)=\sqrt{Q}$ and assume that the proportional to the quality
sharing rule is used. Moreover, each player
pays a cost of $1$ per unit of effort, i.e. $c_i(X_i)=X_i$ and where the quality is equal to the effort. The global optimum
is the solution to the unconstrained optimization problem: $\max_{Q\in \R_{+}}\sqrt{Q}-Q$, 
which leads to $Q^*=1/4$ and therefore $SW(Q^*)=1/4$. On the other hand, each player's optimization problem is:
$\max_{q_i \in \R_{+}} q_i\frac{1}{\sqrt{Q}}-q_i$. By symmetry and elementary calculus, we obtain that at the unique Nash equilibrium, the total effort is
$Q = \left(\frac{n-1/2}{n}\right)^2$, and the social welfare is $\frac{2 n-1}{4 n^2}=O(1/n)$.
\qed\end{example}
Linear efforts can lead to inefficiency in a generic class of examples given below.
\begin{proposition}\label{PROP:LINEAR} Assume that the cost function is linear $c(X) = X$, quality is equal to effort and the value is a function $v(Q)$ of the total quality such that there exists $t > 0$ such that $Q>v(Q)$, for every $Q > t$. Then, POA = $\Omega(n)$ for the proportional to quality mechanism.
\end{proposition}

\section{Conclusion and Future Work}
\label{sec:conclusion}
We analyzed a general model of collaboration under uncertainty, capturing settings such as online social computing and scientific co-authorship. We identified 
simple value sharing rules that achieve 
good efficiency in 
a robust manner with respect to informational assumptions. 
%

Some questions remain open for future research. 
We showed that ranking rules, which are highly popular \cite{Arpita2012a,Arpita2012b}, achieve a logarithmic approximation, using fixed-prizes independent of the distribution of qualities (prior-free) and of the game instance. Can a constant approximation be achieved if we allow the fixed prizes associated with each position to depend on the distribution of abilities and on the instance of the game? Also, consider a \emph{two-stage} model where in the first stage players choose the projects to participate in and then
play our collaboration game in the second stage. Can any efficiency guarantee be given on the welfare achieved at the subgame-perfect equilibria of this two-stage game?

\vspace{-.15in}
\bibliographystyle{abbrv}
\bibliography{reward_bib}

\begin{thebibliography}{10}

\bibitem{Anshelevich2010}
E.~Anshelevich and M.~Hoefer.
\newblock {Contribution Games in Networks}.
\newblock {\em Algorithmica}, pages 1--37, 2011.

\bibitem{AS09}
N.~Archak and A.~Sundararajan.
\newblock Optimal design of crowdsourcing contests.
\newblock In {\em ICIS}, 2009.

\bibitem{Bateni2010}
M.~H. Bateni, M.~T. Hajiaghayi, N.~Immorlica, and H.~Mahini.
\newblock {The cooperative game theory foundations of network bargaining
  games}.
\newblock In {\em ICALP}, 2010.

\bibitem{CHS12}
S.~Chawla, J.~D. Hartline, and B.~Sivan.
\newblock Optimal crowdsourcing contests.
\newblock In {\em SODA}, 2012.

\bibitem{DV09}
D.~DiPalantino and M.~{Vojnovi\' c}.
\newblock Crowdsourcing and all-pay auctions.
\newblock In {\em EC}, 2009.

\bibitem{FM99}
E.~Friedman and H.~Moulin.
\newblock Three methods to share joint costs or surplus.
\newblock {\em Journal of Economic Theory}, 87(2):275--312, August 1999.

\bibitem{Arpita2012a}
A.~Ghosh and P.~Hummel.
\newblock Implementing optimal outcomes in social computing: a game-theoretic
  approach.
\newblock In {\em WWW}, 2012.

\bibitem{Arpita2012b}
A.~Ghosh and P.~Hummel.
\newblock Learning and incentives in user-generated content: Multi-armed
  bandits with endogeneous arms.
\newblock In {\em ITCS}, 2013.

\bibitem{GM12}
A.~Ghosh and P.~McAfee.
\newblock Crowdsourcing with endogenous entry.
\newblock In {\em WWW}, 2012.

\bibitem{Goemans2004}
M.~Goemans, L.~E. Li, V.~S. Mirrokni, and M.~Thottan.
\newblock {Market sharing games applied to content distribution in ad-hoc
  networks}.
\newblock {\em ACM MobiHoc}, 2004.

\bibitem{HO01}
A.~Hars and S.~Ou.
\newblock Working for free? - motivations of participating in open source
  projects.
\newblock In {\em Proceedings of the 34th Annual Hawaii International
  Conference on System Sciences ( HICSS-34)-Volume 7 - Volume 7}, HICSS '01,
  pages 7014--, Washington, DC, USA, 2001. IEEE Computer Society.

\bibitem{Hatfield2011}
J.~W. Hatfield and S.~D. Kominers.
\newblock {Multilateral Matching}.
\newblock In {\em ACM EC}, 2011.

\bibitem{JCP09}
S.~Jain, Y.~Chen, and D.~C. Parkes.
\newblock Designing incentives for online question and answer forums.
\newblock In {\em ACM EC}, 2009.

\bibitem{Kleinberg2011}
J.~Kleinberg and S.~Oren.
\newblock {Mechanisms for (mis)allocating scientific credit}.
\newblock STOC, 2011.

\bibitem{Kleinberg}
J.~Kleinberg and E.~Tardos.
\newblock {Balanced outcomes in social exchange networks}.
\newblock STOC, 2008.

\bibitem{Marden2010}
J.~Marden and T.~Roughgarden.
\newblock Generalized efficiency bounds in distributed resource allocation.
\newblock In {\em Decision and Control (CDC), 2010 49th IEEE Conference on},
  pages 2233--2238, 2010.

\bibitem{Meirowitz2003}
A.~Meirowitz.
\newblock On the existence of equilibria to bayesian games with non-finite type
  and action spaces.
\newblock {\em Economics Letters}, 78(2):213--218, 2003.

\bibitem{MW96}
H.~Moulin and A.~Watts.
\newblock Two versions of the tragedy of the commons.
\newblock {\em Review of Economic Design}, 2(1):399--421, December 1996.

\bibitem{RLTRKR06}
A.~M. Rashid, K.~Ling, R.~D. Tassone, P.~Resnick, R.~Kraut, and J.~Riedl.
\newblock Motivating participation by displaying the value of contribution.
\newblock In {\em Proceedings of the SIGCHI Conference on Human Factors in
  Computing Systems}, CHI '06, pages 955--958, New York, NY, USA, 2006. ACM.

\bibitem{Rosen1965}
J.~B. Rosen.
\newblock {Existence and uniqueness of equilibrium points for concave n-person
  games}.
\newblock {\em Econometrica}, 33(3):520--534, 1965.

\bibitem{Roughgarden2009}
T.~Roughgarden.
\newblock {Intrinsic robustness of the price of anarchy}.
\newblock In {\em STOC}, 2009.

\bibitem{Roughgarden2012}
T.~Roughgarden.
\newblock The price of anarchy in games of incomplete information.
\newblock In {\em ACM EC}, 2012.

\bibitem{Roughgarden2010}
T.~Roughgarden and F.~Schoppmann.
\newblock {Local smoothness and the price of anarchy in atomic splittable
  congestion games}.
\newblock In {\em SODA}, 2011.

\bibitem{S06}
S.~K. Shah.
\newblock Motivation, governance, and the viability of hybrid forms in open
  source software development.
\newblock {\em Manage. Sci.}, 52(7):1000--1014, July 2006.

\bibitem{Syrgkanis2012}
V.~Syrgkanis.
\newblock Bayesian games and the smoothness framework.
\newblock {\em Arxiv}, 1203.5155, 2012.

\bibitem{Vetta2002}
A.~Vetta.
\newblock {Nash equilibria in competitive societies, with applications to
  facility location, traffic routing and auctions}.
\newblock 2002.

\bibitem{KSL03}
G.~von Krogh, S.~Spaeth, and K.~R. Lakhani.
\newblock Community, joining, and specialization in open source software
  innovation: a case study.
\newblock {\em Research Policy}, 32(7):1217 -- 1241, 2003.
\newblock <ce:title>Open Source Software Development</ce:title>.

\bibitem{YK03}
Y.~Ye and K.~Kishida.
\newblock Toward an understanding of the motivation of open source software
  developers.
\newblock In {\em Software Engineering, 2003. Proceedings. 25th International
  Conference on}, pages 419--429, 2003.

\end{thebibliography}

\newpage
\begin{appendix}

\section{Sharing Rules that Satisfy the Marginal Contribution}\label{sec:sharing-rules}

\begin{lemma} The proportional sharing rule satisfies the marginal contribution condition when the value functions are increasing and concave functions
of the sum of submitted qualities and the value with no submissions is $0$.
\label{thm:poa_total_concave}
\end{lemma}
\begin{proof}
Let $Q_j = \sum_{i\in N_j}q_i^j$ and let $v_j(q^j)=\tilde{v}_j(Q^j)$. Since $\tilde{v}_j(Q^j)$ is concave and $\tilde{v}_j(0)=0$ then for any $y\in [0,Q_j]: v_j(y) \geq
y\frac{\tilde{v}_j(Q_j)}{Q_j}$. By setting $y=Q_j-q_i^j$ we obtain:
  $$u_i^j(q^j)=q_i^j\frac{\tilde{v}_j(Q_j)}{Q_j}\geq \tilde{v}_j(Q^j)-\tilde{v}_j(Q_j-q_i^j)= v_j(q^j)-v_j(q_{-i}^j)$$
 \qed\end{proof}

\begin{lemma} When $v_j(q^j)=\max_{i\in N_j}q_i^j$, then awarding all the value to the highest quality player, satisfies the marginal contribution condition.
\end{lemma}
\begin{proof}
The marginal contribution of any player other than the highest one, is $0$. The marginal contribution of the highest quality player is at most the total value.
\qed\end{proof}

\begin{lemma}\label{lem:marginal-contr} Sharing proportional to the marginal contribution: 
\begin{equation*}
u_i^j(q_i^j) = \frac{v_j(q^j)-v_j(q_{-i}^j)}{\sum_{k\in N_j} (v_j(q^j)-v_j(q_{-k}^j))}v_j(q^j)
\end{equation*}
satisfies the marginal contribution condition.
\end{lemma}
\begin{proof}
By the submodularity of the value function we can show that 
\begin{equation*}
v_j(q^j)\geq\sum_{k\in N_j} (v_j(q^j)-v_j(q_{-k}^j))
\end{equation*}
Consider
adding the players sequentially and summing the marginal contribution of a player to the value at the time when he was added. 
This summation will be equal to the final value of the project. Additionally, observe that the marginal contribution of a player
at the time that he was added is greater than his marginal contribution to the final value, by submodularity.
\qed\end{proof}

\begin{lemma}
Sharing according to the Shapley value satisfies the marginal contribution condition. 
\end{lemma}
\begin{proof}
The Shapley value of a player is defined as follows: consider a random permutation of the players and consider adding the players
sequentially according to this random permutation. The Shapley value of a player is his expected marginal contribution at the time
that he is added (in expectation over all permutations). 

Observe that by submodularity for any permutation, the marginal contribution of a player at the time that he is added is
at least his final marginal contribution to the value. Thus for any permutation a player is awarded at least his marginal contribution
and therefore, in expectation he is awarded at least his marginal contribution.
\qed\end{proof}

%
%

\section{Local Smoothness for Utility Maximization Games}\label{sec:local-smoothness}

We will present the local smoothness framework briefly, adapted to a utility
maximization problem instead of a cost minimization. We also slightly generalize
the framework by requiring that the player cost functions be continuously
differentiable almost everywhere only at equilibrium points. 

Consider a utility maximization game $\langle N,(S_i)_{i\in N},(u_i)_{i \in
N} \rangle$, where each player's strategy space is a continuous convex subset of
a Euclidean space $\R^{m_i}$. Let $u_i:\times_{i\in N} S_i \rightarrow \R$ be
the utility function for a player $i$, which is assumed to be
continuous and concave in $x_i\in S_i$ for each fixed value of $x_{-i}\in S_{-i}$. The concavity and continuity 
assumptions lead to existence of a pure Nash equilibrium using the classical result of Rosen \cite{Rosen1965}.
Also for a strategy profile $x\in \times_{i\in N} S_i$ we denote with
$SW(x)=\sum_{i\in N}u_i(x)$.

\begin{definition}[Local Smoothness \cite{Roughgarden2010}] A utility maximization game with convex strategy spaces is locally $(\lambda,\mu)$-smooth with respect to a strategy profile $x^*$ iff for every strategy profile $x$ at which
$u_i(x)$ are continuously differentiable:
$$ \sum_{i\in N} [ u_i(x)+\nabla_i u_i(x)\cdot(x^*_i-x_i) ] \geq
\lambda SW(x^*) - \mu SW(x)$$
where $\nabla_i u_i \equiv \left(\frac{\partial u_i}{\partial x_i^1},\ldots, \frac{\partial
u_i}{\partial x_i^{m_i}}\right)$.
\end{definition}
%

\begin{theorem}[\cite{Roughgarden2010}] If a utility maximization game is locally $(\lambda,\mu)$-smooth
with respect to a strategy profile $x^*$ and $u_i(x)$ is continuously
differentiable at every pure Nash equilibrium of the
game then the social welfare at equilibrium is at least$\frac{\lambda}{1+\mu} SW(x^*)$
\end{theorem}
\begin{proof}
 The proof is an adaptation of \cite{Roughgarden2010} for the case of
maximization games. In addition, we generalize to the case of differentiability only at equilibrium, rather than everywhere.  The latter
extension is essential for the class of games that we study. 

The key claim is that if $x$ is a pure Nash equilibrium
then 
$$ \nabla_i u_i(x)\cdot (x_i^*-x_i) \leq 0 $$

Given this claim then we obtain the theorem, since:
$$ SW(x)\geq \sum_{i\in N} u_i(x)+ \nabla_i u_i(x)\cdot
(x_i^*-x_i) \geq \lambda SW(x^*)-\mu SW(x)$$

To prove the claim define $x^{\epsilon} = ((1-\epsilon)x_i+\epsilon
x_i^*,x_{-i})$ and observe that since $u_i$ is differentiable at equilibrium: $\lim_{\epsilon \rightarrow 0}
\frac{1}{\epsilon}(u_i(x^{\epsilon})-u_i(x)) = \nabla_i
u_i(x)\cdot (x_i^*-x_i)$. Thus if $\nabla_i
u_i(x)\cdot (x_i^*-x_i)>0$ then for some $\epsilon_0$ it holds
that $u_i(x^{\epsilon})-u_i(x)>0$, which means that $i$ has a
profitable deviation which contradicts the fact that $x$ is a Nash
Equilibrium.
\qed\end{proof}

\section{Proof of Theorem \ref{THM:FRACTIONAL}}\label{sec:proof-of-fractional}
Now we switch to our collaboration model which is a convex strategy space utility maximization game.
We will also assume that the proportional sharing rule is used. We focus on the case where the project
values are functions of the total submitted quality $Q_j=\sum_{i\in N_j}q_i^j$. 

\subsection{Differentiability at Equilibrium}

To apply the generalized local smoothness framework we first need to show that the utility functions are differentiable at
every pure Nash equilibrium. For the value functions that we consider $v_j(Q_j)=w_j Q_j^{\alpha}$, with $\alpha<1$, the share of each player under the proportional 
sharing mechanism is $u_i^j(q^j) = w_j\frac{q_i^j}{Q_j^{1-\alpha}}$. 

Thus the only point where the utilities can be non-differentiable is at zero. We will 
show that at any equilibrium, $Q_j>\Delta$ for some $\Delta>0$ that depends on the input parameters of the game. Therefore the utilities are differentiable at equilibrium.

\begin{lemma} If $v_j(Q_j)=w_j Q_j^{\alpha}$ for some $\alpha\in (0,1)$ then at any pure Nash equilibrium
of the game $Q_j>0$, for all $j\in M$.
\end{lemma}
\begin{proof}
 Assume that $q$ is an equilibrium where $Q_{j}=0$ for some $j$. Now
consider a player $i$ that participates at $j$. Let $j^*$ be a project on which he has invested 
strictly positive effort (there must exist one by pigeonhole principle). In fact, we can deduce that there
exists a project $j^*$ on which the quality of his submission is at least $t=\min_{j\in M_i} x_i^j(B_i/m)$. 

Consider the deviation where 
he moves a small amount of effort $\epsilon$ from $j^*$ to $j$. This will decrease her quality on $j^*$ by $\delta=\Theta(\epsilon)$
(for sufficiently small $\epsilon$, by Taylor's theorem and by strict monotonicity of $x_i^j(\cdot)$) and will increase her quality on $j$ be $\delta'=\Theta(\epsilon)=\Theta(\delta)$. Let
$\tilde{q}$ be the strategy vector after the deviation.
The increase in his share on project $j$ is going to be $w_j\Theta(\delta^{\alpha})$. Thus the utility difference after the deviation is:
\begin{equation*}
u_i(\tilde{q})-u_i(q)=w_j\cdot \Theta(\delta^{\alpha})-(u_i^{j^*}(q^j) -u_i^j(\tilde{q}^j)) \implies 
\end{equation*}
$$\frac{u_i(\tilde{q})-u_i(q)}{\delta}
=\frac{w_j\cdot \Theta(\delta^{\alpha})}{\delta}-\frac{u_i^{j^*}(q^j) -u_i^j(\tilde{q}^j)}{\delta}$$
As $\delta\rightarrow 0$,
$\frac{q_{j_1}\Theta(\delta^{\alpha})}{\delta}\rightarrow \infty$ and
$$\frac{u_i^{j^*}(q^j) -u_i^j(\tilde{q}^j)}{\delta}\rightarrow
-\frac{\partial u_i^{j^*}(q^{j^*})}{\partial q_i^{j^*}}\geq -\Delta>-\infty$$ 
where $\Delta$ is a lower bounded on the partial derivative of the share of player $i$ when $Q_{j^*}=t$. 

Thus for some $\delta=\delta_0$, small enough, it has to be that $u_i(\tilde{q})>u_i(q)$ and hence $\tilde{q}$ is a profitable deviation for $i$.
\qed\end{proof}

\subsection{An Intermediate Lemma}
In this section we will prove a lemma that applies to any concave function $v_j(Q_j)$ and not necessarily to $Q_j^{\alpha}$. We will use a fact shown 
by
\cite{Roughgarden2010}:
\begin{fact}\label{fact:ineq}
 Let $x_i,y_i \geq 0$ and $x = \sum_i x_i$, $y = \sum_i y_i$ then:
$ \sum_i x_i(y_i-x_i) \leq k(x,y)$
where:
\begin{equation*}
k(x,y) = \begin{cases}
 \frac{y^2}{4} & x\geq y/2 \\
 x(y-x) & x<y/2
\end{cases} 
\end{equation*}
\end{fact}
\begin{lemma}\label{lem:any-concave}
 Assume $v_j(Q_j)$ are concave functions of total quality with $v_j(0)=0$ and are continuously
differentiable at any equilibrium point of the game defined by the proportional
sharing scheme. Let $\bar{v}_j(Q_j)=\frac{v_j(Q_j)}{Q_j}$. Let $\hat{q}$ be a strategy profile. If
for all strategy profiles $q$ at which $\bar{v}(\cdot)$ is differentiable and for all $j\in M$:
$$\hat{Q}_j\bar{v}_j(Q_j)+\bar{v}'_j(Q_j)k(Q_j,\hat{Q}_j)\geq \lambda\cdot \hat{Q}_j\cdot \bar{v}_j(\hat{Q}_j)-\mu\cdot
Q_j\cdot \bar{v}_j(Q_j)$$
then the game is locally $(\lambda,\mu)$-smooth with respect to $x^*$.
\label{thm:concave}
\end{lemma}
\begin{proof}
 Let $q$ be some strategy profile and  $\hat{q}$ some other outcome. From the definition of the
proportional sharing scheme: $$u_i(q) = \sum_{j\in M_i}q_i^j \frac{v_j(Q_j)}{Q_j} = \sum_{j\in M_i}
q_i^j\bar{v}_j(Q_j)$$

$$\nabla_i
u_i(q)\cdot (\hat{q}_i-q_i) = \sum_{j\in M_i} \frac{\partial
(q_i^j\bar{v}_j(Q_j))}{\partial q_i^j}(\hat{q}_i^j-q_i^j) = 
\sum_{j\in M_i} (\bar{v}_j(Q_j) + q_i^j \bar{v}'_j(Q_j))(\hat{q}_i^j-q_i^j)$$
Therefore, we have 
\begin{equation*}
 \begin{split}
\sum_{i\in N} u_i(q)+\nabla_i u_i(x)\cdot
(\hat{q}_i-q_i) 
 =~& \sum_{i\in N} \sum_{j\in M_i}q_i^j\bar{v}_j(Q_j)+(\bar{v}_j(Q_j) + q_i^j \bar{v}'_j(Q_j))(\hat{q}_i^j-q_i^j) \\
 =~& \sum_{i\in N} \sum_{j\in M_i}\hat{q}_i^j\bar{v}_j(Q_j)+
\bar{v}'_j(Q_j)q_i^j(\hat{q}_i^j-q_i^j) \\
 =~& \sum_{j\in M} \left( \hat{Q}_j\bar{v}_j(Q_j) + \bar{v}'_j(Q_j) \sum_i [
q_i^j(\hat{q}_i^j-q_i^j) ] \right) \\
 \geq~& \sum_{j\in M} \left( \hat{Q}_j\bar{v}_j(Q_j)+ \bar{v}'_j(Q_j) k(Q_j,\hat{Q}_j)\right) 
 \end{split}
\end{equation*}
where the last inequality follows from Fact \ref{fact:ineq} and the fact that
$\bar{v}_j(Q_j)$ is a non-increasing function (by concavity of $v_j(Q_j)$ and the fact that $v_j(0)=0$), hence $\bar{v}'_j(Q_j)\leq 0$. Combining the last above inequality 
with the assumption of the theorem, we establish the assertion of the theorem.
\qed\end{proof}

\subsection{Establishing the Optimal Bound}

\begin{proofof}{Theorem \ref{thm:fractional-exponents}}
 Utilizing Lemma \ref{lem:any-concave}, we shall compute the value $\hbox{PoA}(\alpha)$ that minimizes $\frac{1+\mu}{\lambda}$ over $\mu,\lambda \geq 0$ subject to
\begin{equation}
y \hat{v}(x) + \hat{v}'(x)k(x,y) \geq \lambda y \hat{v}(y) - \mu x \hat{v}(x) \hbox{ for every } x,y\geq 0
\label{equ:poacond}
\end{equation}
where $\hat{v}(x) = v(x)/x$, $v(x) = w x^{\alpha}$, $w > 0$ and $0< \alpha \leq 1$. Without loss of generality, we can assume $w = 1$. Recall that $k(x,y) = \frac{1}{4}y^2$, for $y/x\leq 2$, and $k(x,y) = x(y-x)$, otherwise. Using this and defining $z = y/x$, it is easily showed that condition (\ref{equ:poacond}) is equivalent to $h_{\lambda,\mu}(z) \geq 0$, for every $z\geq 0$, where
$$
h_{\lambda,\mu}(z) = \left\{
\begin{array}{ll}
z - \frac{1-\alpha}{4}z^2 - \lambda z^\alpha + \mu, & z\leq 2\\
\alpha z - \lambda z^\alpha + \mu, & z > 2. 
\end{array}
\right .
$$
We proceed with considering two different cases.

{\bf Case 1}: $z \leq 2$. In this case, we note that there exist $z_1^0$ and $z_2^0$ such that $z_2^0 > z_1^0 > 0$ and $h_{\lambda,\mu}(z)$ is decreasing on $[0,z_1^0)\cup (z_2^0,\infty)$ and increasing on $(z_1^0,z_2^0)$. Now, note that $h_{\lambda,\mu}'(2) \geq 0$ is equivalent to $\lambda \leq 2^{1-\alpha}$. Hence, if $\lambda \leq 2^{1-\alpha}$, then $\inf_{z\in [0,2]}h_{\lambda,\mu}(z) = h_{\lambda,\mu}(\xi(\lambda))$ where $\xi(\lambda)$ is the smallest positive value $z$  such that $h_{\lambda,\mu}'(z) = 0$. On the other hand, if $\lambda > 2^{1-\alpha}$, we claim that $h_{\lambda,\mu}(z)$ is decreasing on $[0,2]$, and hence $\inf_{z\in [0,2]}h_{\lambda,\mu}(z) = h_{\lambda,\mu}(2) = 1 + \mu + \alpha - 2^{\alpha}\lambda$.

We showed that in the present case, condition (\ref{equ:poacond}) is equivalent to: if $\lambda \leq 2^{1-\alpha}$, the condition is
$$
\mu \geq \xi(\lambda)^\alpha + \frac{1-\alpha}{4}\xi(\lambda)^2 - \xi(\lambda)
$$
otherwise, if $\lambda > 2^{1-\alpha}$, then the condition is
$$
\mu \geq 2^\alpha \lambda - 1 - \alpha.
$$

{\bf Case 2}: $z > 2$. In this case, we note that $h_{\lambda,\mu}(z)$ is a concave function with unique minimum value over positive values at $z = \lambda^{\frac{1}{1-\alpha}}$. Therefore, if $\lambda \leq 2^{1-\alpha}$, then $\inf_{z > 2} h_{\lambda,\mu}(z) = h_{\lambda,\mu}(2) = 1 + \mu + \alpha - 2^\alpha \lambda$, and otherwise, $\inf_{z > 2}h_{\lambda,\mu}(z) = h_{\lambda,\mu}(\lambda^{\frac{1}{1-\alpha}}) = 1 + \mu - \alpha - (1-\alpha)\lambda^{\frac{1}{1-\alpha}}$. Therefore, we have that in the present case, condition (\ref{equ:poacond}) is equivalent to: if $\lambda \leq 2^{1-\alpha}$, then the condition is
$$
\mu \geq 2^\alpha \lambda - 1 - \alpha
$$
otherwise, if $\lambda > 2^{1-\alpha}$, the condition is
$$
\mu \geq (1-\alpha)\lambda^{\frac{1}{1-\alpha}}- 1 + \alpha. 
$$

Now note that 
$$
\hbox{PoA}(\alpha) = \hbox{PoA}_1(\alpha) \wedge \hbox{PoA}_2(\alpha)
$$
where
\begin{eqnarray*}
\hbox{PoA}_1(\alpha) &=& \inf_{\lambda \leq 2^{1-\alpha}} \max\left\{\frac{1 + \xi(\lambda)^\alpha + \frac{1-\alpha}{4}\xi(\lambda)^2 - \xi(\lambda)}{\lambda},\frac{2^\alpha \lambda - \alpha}{\lambda}\right\}\\
\hbox{PoA}_2(\alpha) &=& \inf_{\lambda > 2^{1-\alpha}}\max\left\{\frac{2^\alpha\lambda - \alpha}{\lambda},\frac{(1-\alpha)\lambda^{\frac{1}{1-\alpha}} + \alpha}{\lambda}\right\}.
\end{eqnarray*}

For $\hbox{PoA}_2(\alpha)$, the minimum over all positive values of $\lambda$ is at the smallest value of $\lambda$ at which the two functions under the maximum operator intersect and this is at $\lambda = 2^{1-\alpha}$. Therefore, 
$$
\hbox{PoA}_2(\alpha) = \frac{2^\alpha\lambda - \alpha}{\lambda} |_{\lambda = 2^{1-\alpha}} = \frac{2-\alpha}{2^{1-\alpha}}.
$$  
It remains only to show that $\hbox{PoA}_1(\alpha) \geq \hbox{PoA}_2(\alpha)$ and thus $\hbox{PoA}(\alpha) = \hbox{PoA}_2(\alpha)$. It is convenient to use an upper bound for the first term that appears under the maximum operator in the definition of $\hbox{PoA}_1(\alpha)$. To this end, we go back to our analysis of Case~1 and note that $h_{\lambda,\mu}(z) \geq z - \lambda z^\alpha + \mu + \alpha - 1$, for every $0\leq z\leq 2$. Requiring that the right-hand side is greater or equal zero for every $z\in [0,2]$ is a sufficient condition for $h_{\lambda,\mu}(z)\geq 0$ to hold for every $z \in [0,2]$ and it yields
$$
\mu \geq (1-\alpha)[\alpha^{\frac{\alpha}{1-\alpha}}\lambda^{\frac{1}{1-\alpha}} + 1].
$$
We thus have
$$
\hbox{PoA}_1(\alpha) \leq \inf_{\lambda\leq 2^{1-\alpha}}\max\left\{\frac{1+(1-\alpha)[\alpha^{\frac{\alpha}{1-\alpha}}\lambda^{\frac{1}{1-\alpha}} + 1]}{\lambda},\frac{2^\alpha \lambda - \alpha}{\lambda}\right\}.
$$
Now, it is easy to check that the first term under the maximum operator is greater or equal than the second term for every $\lambda \leq 2^{1-\alpha}$, hence
$$
\hbox{PoA}_1(\alpha) \leq \inf_{\lambda\leq 2^{1-\alpha}} \frac{2-\alpha + (1-\alpha)\alpha^{\frac{\alpha}{1-\alpha}}\lambda^{\frac{1}{1-\alpha}}}{\lambda}.
$$
It can be readily checked that the right-hand side is non-increasing and hence the infimum is achieved at $\lambda = 2^{1-\alpha}$ with the value
$$
\frac{2-\alpha + 2(1-\alpha)\alpha^{\frac{\alpha}{1-\alpha}}}{2^{1-\alpha}}
$$
which indeed is greater than or equal to $(2-\alpha)/2^{1-\alpha} = \hbox{PoA}_2(\alpha)$.
\end{proofof}

\section{Proof of Theorem \ref{THM:PRODUCTION-COSTS}}

\begin{proofof}{Theorem \ref{thm:production-costs}}
We show here the second part of the theorem. We will prove that if $\tilde{q}(t)$ is the social welfare
maximizing outcome for type profile $t$ then:
\begin{equation}
(1+\mu)\E_t[C(\tilde{q}(t))]\leq \E_{t}[V(\tilde{q}(t))]
\end{equation}
By the first order conditions, the social welfare maximizing outcome satisfies the constraint that if
$\tilde{q}_i^j(t)>0$ then:
\begin{equation*}
\frac{\partial c_i(X_i(\tilde{q}_i(t_i))}{\partial q_i^j}=\frac{\partial v_j(\tilde{Q}_j(t))}{\partial q_i^j}\implies 
c_i'(X_i(\tilde{q}_i(t_i)) = \frac{1}{(x_i^j(\tilde{q}_i^j(t)))'}(v_j(\tilde{Q}_j))'
\end{equation*}
Thus we observe that all projects in which a player puts positive effort
the right hand side is identical. Using the above property and the convexity of $x_i^j(q_i^j)$ we obtain:
\begin{align*}
X_i(\tilde{q}_i(t_i))\cdot c_i'(X_i(\tilde{q}_i(t_i))=~&\sum_{j\in M_i}x_i^{j}(\tilde{q}_i^j)\cdot c_i'(X_i(\tilde{q}_i(t_i))=
\sum_{j\in M_i} \frac{x_i^{j}(\tilde{q}_i^j)}{(x_i^j(\tilde{q}_i^j(t)))'}(v_j(\tilde{Q}_j))'\\
\leq~& \sum_{j\in M_i}q_i^{j}\cdot (v^{j}(\tilde{Q}_{j}))'
\end{align*}
Summing over all players and using the lower bound on the elasticity of the cost functions we obtain:
\begin{align*}
(1+\mu)C(\tilde{q}(t))\leq~& \sum_{i\in N} X_i(\tilde{q}_i(t_i))\cdot c_i'(X_i(\tilde{q}_i(t_i))=\sum_{i\in N}\sum_{j\in M_i}\tilde{q}_i^j (v^{j}(\tilde{Q}_{j}))'\\
=~&\sum_{j\in M} \tilde{Q}_j (v_j(\tilde{Q}_j))'\leq \sum_{j\in M} v_j(\tilde{Q}_j) = V(\tilde{q})
\end{align*}
Now using Theorem \ref{thm:general_effort_costs} we obtain that for any Nash Equilibrium:
\begin{align*}
2\E_{t}[V(q(t))]\geq~& \E_{t}[SW^t(q(t))+V(q(t))]\geq \E_{t}[SW^t(\tilde{q})]=\E_{t}[V(\tilde{q}(t))-C(\tilde{q}(t))\\
\geq~& \frac{\mu}{\mu+1}\E_{t}[V(\tilde{q}(t))]
\end{align*}
\end{proofof}

\section{Proof of Proposition \ref{PROP:LINEAR}}

\begin{proofof}{Proposition \ref{PROP:LINEAR}} 
Social welfare is given by $\mathrm{SW}(x) = v(X) - \sum_{i=1}^n c(x_i)$ and socially optimal allocation is such that $x_i = X/n$ for every player $i$, and 
\begin{equation}
v'(X) - c'(X/n) = 0.
\label{equ:foc1}
\end{equation}
Due to the symmetry, we use the notation $\mathrm{SW}(X) = v(X) - nc(X/n)$.

Furthermore, a Nash equilibrium allocation $x$ is such that for every $i$, $x_i$ maximizes
$$
\frac{x_i}{X}v(X)  - c(x_i).
$$
The first order optimality condition reads as, for every $i$,
$$
\frac{v(X)}{X} + x_i \frac{v'(X)X - v(X)}{X^2} - c'(x_i) = 0.
$$
Therefore, $x_i = X/n$ for every $i$, and
\begin{equation}
v'(X) - c'(X/n) + \left(1-\frac{1}{n}\right)\frac{v(X) - v'(X)X}{X} = 0.
\label{equ:foc2}
\end{equation}

It is not difficult that the Nash equilibrium of the game is optimal allocation for a (virtual) social welfare function defined as follows
$$
\hat{\mathrm{SW}}(X) = \frac{1}{n}v(X) + \left(1-\frac{1}{n}\right)\int_0^X \frac{v(y)}{y}dy - nc(X/n).
$$
As an aside remark, note since $v(X)$ is a concave function, $v(X)/X$ is non-increasing over $X\geq 0$ and hence, $\mathrm{SW}(X) \leq \hat{\mathrm{SW}}(X)$, for $x\geq 0$.

In view of the identity (\ref{equ:foc2}), we have
\begin{eqnarray}
\mathrm{SW}(X) &=& v(X) - nc(X/n)\nonumber\\
&=& v(X) - nc(X/n) -X \left(v'(X) - c'(X/n) + \left(1-\frac{1}{n}\right)\frac{v(X) - v'(X)X}{X}\right)\nonumber\\
&=& \frac{v(X)-v'(X)X}{n} + n\left((X/n)c'(X/n)-c(X/n)\right).\label{equ:sw}
\end{eqnarray}

We observe that
$
(X/n)c'(X/n)-c(X/n) = 0
$
if and only if $c(X)$ is a linear function and in this case at Nash equilibrium, it holds
$$
\mathrm{SW}(X) = \frac{v(X)-v'(X)X}{n}.
$$

From (\ref{equ:foc1}) we have that optimum social welfare is a constant, independent of $n$. From (\ref{equ:foc2}) and the fact $X\leq t$, we have that in Nash equilibrium
$$
\mathrm{SW}(X) \leq \frac{\max_{X\in [0,t]}\{v(X) - v'(X)X\}}{n} = \frac{v(t) - v'(t)t}{n} = O(1/n).
$$
The result follows.
\end{proofof}

\begin{remark} We observe that socially optimal and Nash equilibrium allocations are determined by (\ref{equ:foc1}) and (\ref{equ:foc2}), respectively, where for asymptotically large $n$, the behaviors of the functions $v'(x)$ for large $x$ and $c'(x)$ for small $x$ play a key role.
\end{remark}

\section{Importance of Monotonicity} 

Throughout the paper we assumed that the value produced is monotone in the quality of the submissions. While being a natural assumption for most of our applications
one can think of collaborative settings where more effort is not always better. What can we say about such non-monotone situations? 

We show that without the monotonicity assumption then the marginal contribution condition is not sufficient to guarantee a constant price of anarchy for a generic class of examples and for the simple proportional sharing scheme. To see this consider the case of a single project with value function $v(Q)$ (where
$Q$ is the total effort) that is differentiable, concave, $v(0) = 0$, and is single peaked, i.e. the function is increasing for $0\leq Q < Q^*$ and decreasing for $Q > Q^*$, for some $Q^* > 0$. Without loss of generality, let us assume that $v(Q^*) = 1$ and $v(1) = 0$. We assume that $|v'(1)| < \infty$. In this case, the maximum social welfare is $V(Q^*) = 1$. 

Suppose that each player has a budget of $1$ and that effort is equal to submitted quality. The payoff of player $i$ is given by
$$
u_i(q) = \frac{q_i}{Q}v(Q)
$$
The game has a unique Nash equilibrium at which $\frac{\partial}{\partial q_i} u_i(q) = 0$. By summing this condition over all players we get that
the total quality at equilibrium must satisfy:
$$
\left(1-\frac{1}{n}\right)v(Q) + \frac{1}{n}v'(Q)Q = 0.
$$
Let $\tilde{Q}$ be the solution to the above equation. It is easy note from the last identity that $Q^* < \tilde{Q}\leq 1$. By concavity, we have $v(\tilde{Q}) \geq - v'(\tilde{Q})(1-\tilde{Q})$, which combined with the Nash equilibrium condition yields
$
\tilde{Q} \geq 1 - \frac{1}{n}
$.
Therefore, the price of anarchy is
$$
\frac{V(Q^*)}{V(\tilde{Q})} \geq \frac{1}{v(1-\frac{1}{n})} \geq \frac{1}{-v'(1)} n = \Theta(n).
$$

\end{appendix}

\end{document}